\documentclass[conference]{IEEEtran}
\IEEEoverridecommandlockouts 	
\usepackage{amsmath,amssymb}
\usepackage{amsthm}
\usepackage{color}
\usepackage[mathscr]{eucal}
\usepackage{graphics,graphicx,multicol}
\usepackage{epsfig}
\usepackage{enumerate}
\usepackage{subfigure}
\usepackage{algorithmic}
\usepackage{algorithm}
\usepackage{morefloats}
\usepackage{enumerate}
\usepackage{subfigure}
\usepackage{multirow}
\usepackage{array}
\usepackage{pstricks, pst-node, pst-plot, pst-circ}
\usepackage{moredefs}
\usepackage{cite}
\usepackage{courier}
\newtheorem{thm}{Theorem}[section] 
\newtheorem{lem}[thm]{Lemma}

\usepackage{epstopdf}
\DeclareGraphicsExtensions{.eps}
\begin{document}

\title{An Optimal Application-Aware Resource Block Scheduling in LTE}
\author{Tugba Erpek, Ahmed Abdelhadi, and T. Charles Clancy \\
Hume Center, Virginia Tech, Arlington, VA, 22203, USA\\
\{terpek, aabdelhadi, tcc\}@vt.edu
}
\maketitle

\begin{abstract}
In this paper, we introduce an approach for application-aware resource block scheduling of elastic and inelastic adaptive real-time traffic in fourth generation Long Term Evolution (LTE) systems. The users are assigned to resource blocks. A transmission may use multiple resource blocks scheduled over frequency and time. In our model, we use logarithmic and sigmoidal-like utility functions to represent the users applications running on different user equipments (UE)s. We present an optimal problem with utility proportional fairness policy, where the fairness among users is in utility percentage (i.e user satisfaction with the service) of the corresponding applications. Our objective is to allocate the resources to the users with priority given to the adaptive real-time application users. In addition, a minimum resource allocation for users with elastic and inelastic traffic should be guaranteed. Every user subscribing for the mobile service should have a minimum quality-of-service (QoS) with a priority criterion. We prove that our scheduling policy exists and achieves the maximum. Therefore the optimal solution is tractable. We present a centralized scheduling algorithm to allocate evolved NodeB (eNodeB) resources optimally with a priority criterion. Finally, we present simulation results for the performance of our scheduling algorithm and compare our results with conventional proportional fairness approaches. The results show that the user satisfaction is higher with our proposed method.
\end{abstract}

\pagenumbering{gobble}
\providelength{\AxesLineWidth}       \setlength{\AxesLineWidth}{0.5pt}%
\providelength{\plotwidth}           \setlength{\plotwidth}{8cm}
\providelength{\LineWidth}           \setlength{\LineWidth}{0.7pt}%
\providelength{\MarkerSize}          \setlength{\MarkerSize}{3pt}%
\newrgbcolor{GridColor}{0.8 0.8 0.8}%
\newrgbcolor{GridColor2}{0.5 0.5 0.5}%

\section{Introduction}\label{sec:intro}

The area of resource allocation optimization has received significant interest as the operators face an increasing demand for mobile data traffic. According to Cisco VNI Mobile Forecast Highlights \cite{CiscoVNIForecast}, globally, the mobile data traffic will grow 11-fold from 2013 to 2018 and there will be 4.9 billion mobile users by 2018, up from 4.1 billion in 2013. 

LTE technology offers increased user data rates, improved spectral efficiency and greater flexibility of spectrum usage. The resource scheduling algorithms are not defined in the LTE standard. It is the responsibility of the vendors to implement the optimal algorithms to increase the user satisfaction in a spectrally-efficient way. The mobile application types include both elastic and inelastic traffic. The expectations of the users change depending on the traffic type. Elastic traffic can adjust to wide range of changes in delay and throughput and still meet the user expectations. Inelastic traffic has strict latency and throughput requirements. Real time applications such as VoIP and video streaming are examples of inelastic traffic. It is presented in \cite{YouTubeTraffic} that the video applications such as YouTube's progressive download starts by transferring a significant amount of data in the player's buffer in order to mitigate future re-buffering events. The user running this application should be allocated more bandwidth during this phase.

The key inputs to the resource scheduling process are common. Mainly, two different resource allocation categories can be identified as the opportunistic scheduling and proportional fair scheduling \cite{LTEBook}. It is difficult to ensure fairness and QoS with the opportunistic scheduling. Proportional fair scheduling pays more attention to the user QoS requirements. As a result, there has been increasing research on the proportional fair scheduling algorithms \cite{PropFairScheduling}.

In this paper, we focus on finding the optimal solution for an application-aware resource scheduling problem for LTE systems. We present a survey of related research papers in the following subsection.  

\subsection{Related Work}\label{sec:related}

In \cite{Ahmed_Utility1, Ahmed_Utility2, Ahmed_Utility3}, the authors present optimal rate allocation algorithms for users covered by a single carrier eNodeB. The authors use logarithmic and sigmoidal-like utility functions to represent delay-tolerant and real-time applications, respectively. In \cite{Ahmed_Utility1}, the rate allocation algorithm gives priority to real-time applications over delay-tolerant applications when allocating resources as the utility proportional fairness  rate allocation policy is used. In \cite{Ahmed_Utility2}, the convergence analysis of the single carrier resource allocation algorithm is presented. In \cite{Ahmed_Utility3}, two-stage resource allocation for multi-application users covered by a single carrier is presented.

In \cite{Haya_Utility1}, the authors present multiple-stage carrier aggregation with utility proportional fairness resource allocation algorithm. The users allocate the resources from the primary carrier eNodeB until all the resources in the eNodeB are allocated. The users switch to the secondary carrier eNodeB to allocate more resources, and so forth. In \cite{Haya_Utility2}, spectrum sharing of public safety and commercial LTE bands is assumed. The authors presented a resource allocation algorithm with priority given to the public safety users. In \cite{Ahmed_Utility4}, authors present a joint carrier aggregation resource allocation where the allocated rates are optimal. In \cite{MultiBand}, a distributed solution of resource allocation for proportional fairness is provided for multi-band wireless systems. The proposed approach is not specific to the LTE systems. In \cite{SelfOrganizedLTE}, a distributed protocol that aims to achieve weighted proportional fairness among UEs for LTE systems is presented. A resource block scheduling problem is formulated as a convex optimization problem. A weighted proportional fairness among all the UEs is achieved by setting a priori priority indicator. The weights play an important role while solving the optimization problem however the optimal resource scheduling is not guaranteed since the weights are set a priori. All the UE applications are treated the same whenever the initial weights are set equal.   

\subsection{Our Contributions}\label{sec:contributions}
Our contributions in this paper are summarized as:
\begin{itemize}
\item We introduce an application-aware scheduling scheme that involves users with real-time and delay-tolerant applications. The proposed scheduling scheme gives priority to real-time application users while allocating resource blocks. 
\item We prove that the proposed scheduling scheme exists and is optimal.
\end{itemize}

The remainder of this paper is organized as follows. Section \ref{sec:Problem_formulation} presents the system model and problem formulation. Section \ref{sec:Proof} proves the global optimal solution exists and is tractable. In Section \ref{sec:Algorithm}, we present our centralized resource block scheduling algorithm for the utility proportional fairness optimization problem. Section \ref{sec:sim} discusses simulation setup and provides quantitative results along with discussion. Section \ref{sec:conclude} concludes the paper.

\section{System Model and Problem Setup}\label{sec:Problem_formulation}
LTE downlink transmission resources has time, frequency and space dimensions. Using multiple transmit and receive antennas provide the spatial dimension. The frequency dimension is divided to subcarriers. The time dimension is first divided to 10 ms radio frames. Frames are further subdivided into ten 1 ms subframes, each of which is split into two 0.5 ms slots \cite{LTEBook}. 

The smallest unit of resource is the resource element. A resource element consists of one subcarrier for a duration of one OFDM symbol. A resource block is comprised of 12 continuous subcarriers. It has 84 resource elements in the case of the normal cyclic prefix length, and 72 resource elements in the case of the extended cyclic prefix. Both time-division multiplexing and frequency-division multiplexing can be achieved with LTE systems. We focus on the latter in this paper. 

A resource block can be allocated to only one user for reuse-1 radio systems. For a centralized resource block scheduling algorithm, the eNodeB decides which UE will be allocated for each resource block. We use the same problem setup as in \cite{SelfOrganizedLTE}. Without loss of generality, we define $B$ to be the set of eNodeBs, $M$ to be the set of UEs and $Z$ to be the set of resource blocks. We use $z \in Z$ to denote a single resource block. The total throughput allocated by the eNodeB to the $i^{th}$ UE over all the resource blocks is given by $r_i$. Each UE has its own utility function $U_i(r_i)$ that corresponds to the type of traffic being handled by the UE. Our objective is to determine which resource blocks the eNodeB should allocate to each UE. 

\subsection{User Throughput}\label{sec:throughput}

We denote the throughput of UE $i$ on resource block $z$ when it is scheduled by eNodeB $b(i)$ as $H_{i,b(i),z}$. In each frame, eNodeB $b(i)$ schedules one UE in each of the resource blocks in the frame. Let $\phi_{i,b(i),z}$ be the proportion of the frames that UE $i$ is scheduled by eNodeB $b(i)$ in resource block $z$. The overall throughput of UE $i$, which is the sum of its throughput over all the resource blocks can be written as:
\begin{equation}\label{eqn:throughput}
r_i = \sum_{z\epsilon Z}\phi_{i,b(i),z} H_{i,b(i),z} 
\end{equation}

\subsection{Utility Functions}\label{sec:utilities}

We use the same utility functions as in \cite{Ahmed_Utility1}. We assume all user utilities $U_i(r_i)$ in this model are strictly concave or sigmoidal-like functions. The utility functions have the following properties: 
\begin{itemize}
\item $U_i(0) = 0$ and $U_i(r_i)$ is an increasing function of $r_i$.
\item $U_i(r_i)$ is twice continuously differentiable in $r_i$.
\end{itemize}
In our model, we use the normalized sigmoidal-like utility function, as in \cite{DL_PowerAllocation}, that can be expressed as 
\begin{equation}\label{eqn:sigmoid}
U_i(r_i) = c_i\Big(\frac{1}{1+e^{-a_i(r_i-b_i)}}-d_i\Big)
\end{equation}
where $c_i = \frac{1+e^{a_ib_i}}{e^{a_ib_i}}$ and $d_i = \frac{1}{1+e^{a_ib_i}}$. So, it satisfies $U(0)=0$ and $U(\infty)=1$. In Figure \ref{fig:SigLogUtility}, the normalized sigmoidal-like utility function with $a=5$ and $b=10$ is a good approximation for a step function (e.g. VoIP), and $a=0.5$ and $b=20$ is a good approximation to an adaptive real-time application (e.g. video streaming). Additionally, we use the normalized logarithmic utility function, as in \cite{UtilityFairness}, that can be expressed as 
\begin{equation}\label{eqn:log}
U_i(r_i) = \frac{\log(1+k_ir_i)}{\log(1+k_i r_{max})}
\end{equation}
where $r_{max}$ is the required rate for the user to achieve 100\% utility percentage and $k_i$ is the rate of increase of utility percentage with the allocated rate $r_i$. So, it satisfies $U(0)=0$ and $U(r_{max})=1$. The logarithmic utility functions with $k=15$ and $k=0.1$ are also shown in Figure \ref{fig:SigLogUtility} representing the delay tolerant traffic.

\begin{figure}
    \centering
    \includegraphics[width=3.5in]{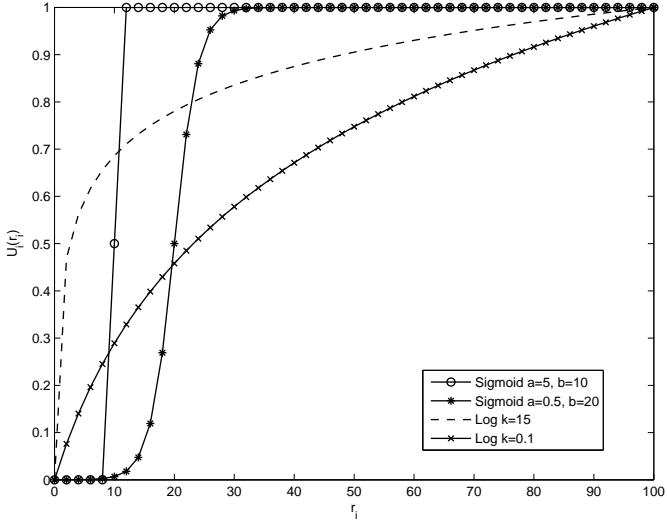}
    \caption{The sigmoidal-like and logarithmic utility functions.}
    \label{fig:SigLogUtility}
\end{figure} 

\subsection{Scheduling Problem}\label{sec:scheduling}

We consider the utility proportional fairness objective function given by 
\begin{equation}\label{eqn:utility_fairness}
\begin{aligned}
& \underset{\textbf{$\phi_{i,b(i),z}$}} {\text{max}}
& & \prod_{i=1}^{M}U_i(\sum_{z\epsilon Z}\phi_{i,b(i),z} H_{i,b(i),z}) \\
\end{aligned}
\end{equation}
where $M$ is the number of UEs in the coverage area of the eNodeB. The goal of this resource scheduling objective function is to allocate the resources for each UE that maximizes the total cellular network objective (i.e. the product of the utilities of all the UEs) while ensuring proportional fairness between individual utilities. This resource scheduling objective function ensures non-zero resource allocation for all users. Therefore, the corresponding resource scheduling optimization problem guarantees minimum QoS for all users. In addition, this approach allocates more resources to users with real-time applications providing improvement to the QoS of LTE system.

The basic formulation of the utility proportional fairness resource scheduling problem is given by the following optimization problem:
\begin{equation}\label{eqn:opt_prob_fairness}
\begin{aligned}
& \underset{\textbf{$\phi_{i,b(i),z}$}} {\text{max}}
& & \prod_{i=1}^{M}U_i(\sum_{z\epsilon Z}\phi_{i,b(i),z} H_{i,b(i),z}) \\
& \text{subject to}
& & \sum_{i=1}^{M}\phi_{i,b(i),z} =1\\
& & &  \phi_{i,b(i),z} \geq 0, \;\;\;\;\; i = 1,2, ...,M.
\end{aligned}
\end{equation}
We prove in Section \ref{sec:Proof} that there exists a tractable global optimal solution to the optimization problem (\ref{eqn:opt_prob_fairness}).

\section{The Global Optimal Solution}\label{sec:Proof}

In the optimization problem (\ref{eqn:opt_prob_fairness}), since the objective function $\arg \underset{\textbf{$\phi_{i,b(i),z}$}} \max \prod_{i=1}^{M}U_i(\sum_{z\epsilon Z}\phi_{i,b(i),z} H_{i,b(i),z})$ is equivalent to $\arg \underset{\textbf{$\phi_{i,b(i),z}$}} \max \sum_{i=1}^{M}\log(U_i(\sum_{z\epsilon Z}\phi_{i,b(i),z} H_{i,b(i),z}))$, then optimization problem (\ref{eqn:opt_prob_fairness}) can also be written as:

\begin{equation}\label{eqn:opt_prob_fairness_mod}
\begin{aligned}
& \underset{\textbf{$\phi_{i,b(i),z}$}} {\text{max}}
& & \sum_{i=1}^{M}\log(U_i(\sum_{z\epsilon Z}\phi_{i,b(i),z} H_{i,b(i),z})) \\
& \text{subject to}
& & \sum_{i=1}^{M}\phi_{i,b(i),z} =1\\
& & &  \phi_{i,b(i),z} \geq 0, \;\;\;\;\; i = 1,2, ...,M.
\end{aligned}
\end{equation}

The utility functions $\log(U_i(\sum_{z\epsilon Z}\phi_{i,b(i),z} H_{i,b(i),z}))$ in the optimization problem (\ref{eqn:opt_prob_fairness_mod}) are strictly concave functions as proved in \cite{Ahmed_Utility1}. As a result, the optimization problem (\ref{eqn:opt_prob_fairness_mod}) is a convex optimization problem and there exists a unique tractable global optimal solution \cite{Ahmed_Utility1}. It follows that the optimization problem (\ref{eqn:opt_prob_fairness}) is also convex. 

Online algorithms are used when the entire input is not available from the start. The input is processed one-by-one in a serial fashion with this approach. Calculating the total throughput of a UE over all resource elements, $r_i$, requires the knowledge of $\phi_{i,b(i),z}$. In \cite{SelfOrganizedLTE}, an online scheduling algorithm is proposed to decrease the computation overhead. We also use an online scheduling algorithm in order to process the throughput information piece-by-piece while solving our optimization problem.
Let $\phi_{i,b(i),z}[k]$ be the proportion of the frames that the resource block $z$ is scheduled for UE $i$ in the first $k$ frames. Then, we can define the proportion of the frames that the resource block $z$ is scheduled for $i$ in the $[k+1]^{th}$ frame as:
\begin{equation*}\label{eqn:online_algorithm}
\begin{aligned}
\phi_{i,b(i),z}[k+1]=
\begin{cases}
	\frac{k-1}{k}\phi_{i,b(i),z}[k]+\frac{1}{k},\\ 
	\text{if UE $i$ is scheduled for $z$}\\
	\text{in $(k+1)^{th}$ frame}\\
	\frac{k-1}{k}\phi_{i,b(i),z}[k],\text{otherwise}\\
\end{cases}
\end{aligned}
\end{equation*}

In our scheduling policy, the eNodeB schedules for the UE that maximizes $\frac{U'_i(\phi_{i,b(i),z})H_{i,b(i),z}}{U_i(\phi_{i,b(i),z})}$ while solving the scheduling problem.\\      

\begin{lem}\label{lem:optimality}  
Using the above scheduling policy, we show that $\lim \inf_{k \rightarrow \infty} \sum \log U_i(\sum_z\phi_{i,b(i),z}H_{i,b(i),z})$ exists for optimization problem (\ref{eqn:opt_prob_fairness}).
\end{lem}

\begin{proof}
We define $L(\phi)=\sum_{i=1}^M \log U_i(\sum_z\phi_{i,b(i),z}H_{i,b(i),z})$ where $\phi$ is the short form of $\phi_{i,b(i),z}$ and $\phi[k]$ is the short form of $\phi_{i,b(i),z}[k]$. Using Taylor's theorem, for any $\phi$ and $\Delta\phi\\$

$L(\phi+\Delta\phi)=L(\phi)+L'(\phi)\Delta\phi+\pi(\phi,\Delta\phi)\\$ 

where $|\pi(\phi+\Delta\phi)|<a |\Delta\phi|^2$, for some constant $a$. 

Let $\Delta\phi_{i,b(i),z}[k]=\phi_{i,b(i),z}[k+1]-\phi_{i,b(i),z}[k]$, then 
\begin{equation*}\label{eqn:delta_phi_eqn}
\begin{aligned}
\Delta\phi_{i,b(i),z}[k]=
\begin{cases}
	\frac{1}{k}-\frac{\phi_{i,b(i),z}[k]}{k},\\ 
	\text{if UE $i$ is scheduled for $z$ in $(k+1)^{th}$ frame}\\
	\frac{-\phi_{i,b(z),z}[k]}{k},\text{otherwise}\\
\end{cases}
\end{aligned}
\end{equation*}
 
$|\Delta\phi_{i,b(i),z}[k]|<\frac{1}{k}$, for all $i$ and $z$. As a result;

\begin{equation*}\label{eqn:L_inequality}
\begin{aligned}
L(\phi[k+1]) & = L(\phi[k]+\Delta\phi[k]),\\
& \geq L(\phi[k])+\Delta 
L(\phi[k])-\frac{a}{k^2},\\
&=L(\phi[k])+ \Big( \sum_{b \in B,z \in Z}\sum_i \frac{U'_i(\phi_{i,b(i),z})H_{i,b(i),z}}{U_i(\phi_{i,b(i),z})}\\
&\Delta\phi_{i,b(i),z[k]} \Big) -\frac{a}{k^2} \\
&=L(\phi[k])+\frac{1}{k}\sum_{b \in B,z \in Z}\Big (\max_i\frac{U'_i(\phi_{i,b(i),z})H_{i,b(i),z}}{U_i(\phi_{i,b(i),z})} \\
&-\sum_i\frac{U'_i(\phi_{i,b(i),z})H_{i,b(i),z}}{U_i(\phi_{i,b(i),z})}\Big )-\frac{a}{k^2}
\end{aligned}
\end{equation*}
Since $\sum_i\phi_{i,b(i),z[k]}=1$ based on (\ref{eqn:opt_prob_fairness}), the last equation can be written as 
$L(\phi[k]+\Delta\phi[k]) \geq L(\phi[k])-\frac{a}{k^2} \\$.

Let $\beta := \lim \sup_{k \rightarrow \infty} L(\phi[k])$. For any $\epsilon > 0$, there exists large enough $K$ so that $L(\phi[K])>\beta-\frac{\epsilon}{2}$ and $\sum_{k=K}^{\infty}\frac{a}{k^2}<\frac{\epsilon}{2}$. For any $\hat{k}>K$, $L(\phi[\hat{k}]) \geq L(\phi[K])-\sum_{k=K}^{\hat{k}}\frac{a}{k^2} > \beta-\epsilon$. Therefore, $L(\phi[k])$ converges to $\beta$, as $k \rightarrow \infty$. \\  
   
Due to the constraint $\sum_{i=1}^M \phi_{i,b(i),z} = 1$ in (\ref{eqn:opt_prob_fairness_mod}), $\phi$ is a solution to the optimization problem if and only if 

\begin{equation}\label{eqn:opt_Maximization}
\begin{aligned}
\frac{dL}{d\phi_{i,b(i),z}} & =\frac{U'_i(\phi_{i,b(i),z})H_{i,b(i),z}}{U_i(\phi_{i,b(i),z})} \\ 
& =\max_j\frac{U'_j(\phi_{j,b(j),z})H_{j,b(j),z}}{U_j(\phi_{j,b(j),z})} 
\end{aligned}
\end{equation}
for all $i$ and $z$ such that $\sum_{i=1}^{M}\phi_{i,b(i),z}=1$ and $\phi_{i,b(i),z} \geq 0$.

\end{proof}

\begin{thm}\label{thm:limitmax}
Using the scheduling policy (\ref{eqn:opt_Maximization}), $\lim \inf_{k \rightarrow \infty} \sum \log U_i(\sum_z\phi_{i,b(i),z}H_{i,b(i),z})$ achieves the maximum of the optimization problem (\ref{eqn:opt_prob_fairness}). 
\end{thm}
\begin{proof}
Suppose $\lim_{k \rightarrow \infty}L(\phi[k])$ does not achieve the maximum of the optimization problem. There exists $\delta >   0$, $\lambda > 0$, and positive integer $K$ such that for all $k > K$, there exists some $i_k \in M$ and $z_k \in Z$ so that $\phi_{i_k,b(i_k),z_k}[k]>\delta$ and  $\frac{U'_{i_k}(\phi_{i_k,b(i_k),z_k})H_{i_k,b(i_k),z_k}}{U_i(\phi_{i_k,b(i_k),z_k})}=\max_{j:b(j)=b(i_k)}\frac{U'_j(\phi_{j,b(j),z_k})H_{j,b(j),z_k}}{U_j(\phi_{j,b(j),z_k})}-\lambda$. At this point, we have 
\begin{equation*}\label{eqn:L_maximize}
\begin{aligned}
& L(\phi[k+1])-L(\phi[k]) \geq L'(\phi[k]))\Delta \phi[k]-\frac{a}{k^2}\\
& =\sum_{b \in B, z \in Z}\sum_{i}\frac{U'_i(\phi_{i,b(i),z}[k])H_{i,b(i),z}}{U_i(\phi_{i,b(i),z}[k])}\Delta\phi_{i,b(i),z}[k]-\frac{a}{k^2}\\
& =\frac{\delta\lambda}{k}-\frac{a}{k^2} \geq \frac{\delta\lambda}{2k},
\end{aligned}
\end{equation*}
for large enough $k$. Since $\sum_{k=1}^\infty \frac{1}{k}=\infty$, which is a contradiction. As a result, $\lim_{k \rightarrow \infty}L(\phi[k])$ achieves the maximum of the optimization problem.
\end{proof}


\section{Centralized Optimization Algorithm}\label{sec:Algorithm}
Our centralized resource scheduling algorithm allocates resources with utility proportional fairness, which is the objective of our problem formulation. The eNodeB allocates the resource block $z$ for the UE that has the maximum $U'(\phi_{i,b(i),z})H_{i,b(i),z}/ U(\phi_{i,b(i),z})$. Since the optimization problem is solved using the utility functions the priority will be given to the sigmoidal functions which have more strict delay and throughput requirements. 

Algorithm (\ref{alg:eNodeB}) shows our resource scheduling algorithm. This algorithm allocates resources with utility proportional fairness, which is the objective of the problem formulation. The eNodeB runs the algorithm and makes resource scheduling decisions.  

\begin{algorithm}
\caption{Resource Scheduling Algorithm}\label{alg:eNodeB}
\begin{algorithmic}
\STATE {$\phi_{i,b(i),z}=0$; $r_i=0$}
\FOR {$z=1 \rightarrow |Z|$}
	\STATE {Estimate the channel gain $H_{i,b(i),z}$}
	\IF {$l = \arg \max_j\frac{U'_j(\phi_{j,b(j),z})H_{j,b(j),z}}{U_j(\phi_{j,b(j),z})}$} 
	\STATE {$\phi_{l,b(l),z}[k+1]=\frac{k-1}{k}\phi_{l,b(l),z}[k]+\frac{1}{k}$} \\
	\COMMENT {Resource block $z$ allocated to UE $l$}   
	\STATE {$\phi_{i,b(i),z}[k+1]=\frac{k-1}{k}\phi_{i,b(i),z}[k]$} \\
	\COMMENT {For $i \neq l$}
\ENDIF 
\ENDFOR
\end{algorithmic}
\end{algorithm}

\section{Simulation Results}\label{sec:sim}

In this section, we present and compare the simulation results for both the application-aware resource scheduling and conventional proportional fairness algorithms. 

We initially present the simulation results of six utility functions corresponding to the UEs shown in Figure \ref{fig:SigLogUtility}. We use three normalized sigmoidal-like functions that are expressed by equation (\ref{eqn:sigmoid}) with different parameters, $a = 5$, $b=10$ which is an approximation to a step function (e.g. VoIP), $a = 3$, $b=20$ which is an approximation of an adaptive real-time application (e.g. standard definition video streaming), and $a = 1$,  $b=30$ which is also an approximation of an adaptive real-time application (e.g. high definition video streaming). We use three logarithmic functions that are expressed by equation (\ref{eqn:log}) with $r_{max}=100$ and different $k_i$ parameters which are approximations for delay tolerant applications (e.g. FTP). We use $k =\{15, 3, 0.5\}$. The simulation was run in MATLAB. The algorithm in (\ref{alg:eNodeB}) was applied to the logarithmic and sigmoidal-like utility functions listed above. Unity channel gain is assumed for each UE. The Quality of Experience (QoE) is calculated for each UE when the utility proportional fairness approach is used. The results are shown in subplot one of Figure \ref{fig:sim:SubPlots}. The bandwidth allocated for the sigmoidal functions are higher since sigmoidal-like utility functions have priority over the logarithmic utility functions. The QoE is above $50\%$ for all the users. 

\begin{figure}
    \centering
    \includegraphics[width=3.5in]{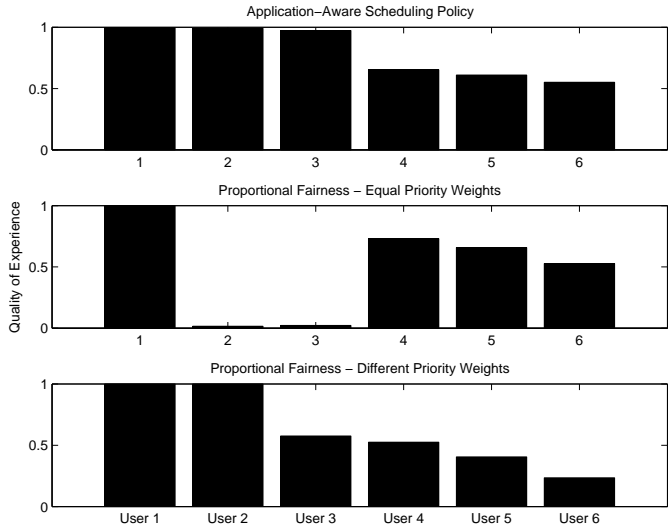}
    \caption{The quality of experience with the application-aware and conventional proportional fairness scheduling policies.}
    \label{fig:sim:SubPlots}
\end{figure} 


Secondly, we present the simulation results when a conventional proportional fairness approach is followed as in \cite{SelfOrganizedLTE}. We assumed unity channel gain for this simulation, as well. Priority weights, $w_i$, which are user-dependent priority indicators are used in \cite{SelfOrganizedLTE} for each UE. In this algorithm, the weighted proportional fairness is achieved by scheduling the resource block to the UE with the maximum $\frac{w_iH_{i,b(i),z}}{r_i}$. We used two different sets of priority weights in our simulation. Initially, equal weights were applied to each user. This is the worst case scenario when all the applications have the same priority. The QoE is calculated for each user. It is expected that the QoE will be low for the users with real time applications. The results are presented in subplot two of Figure \ref{fig:sim:SubPlots}. The QoE for users 2 and 3 are very close to zero.    


We set the priority weights equal to 10 for the first three users and to 1 for the last three users next. The QoE is calculated again for each user. The results are presented in subplot three of Figure \ref{fig:sim:SubPlots}. The QoE is better compared to equal weight case especially for UE 2 and 3. However, the optimal resource scheduling with maximum throughput is still not achieved since a priori assigned weights are used and the application delay and throughput requirements are not fully taken into account.  


Finally, we compare the results of the maximization problems for all three cases. The objective function values are plotted in Figure \ref{fig:sim:product_comparison}. The results show that the application-aware resource block scheduling increases the total system throughput and user satisfaction. 

\begin{figure}
    \centering
    \includegraphics[width=3.5in]{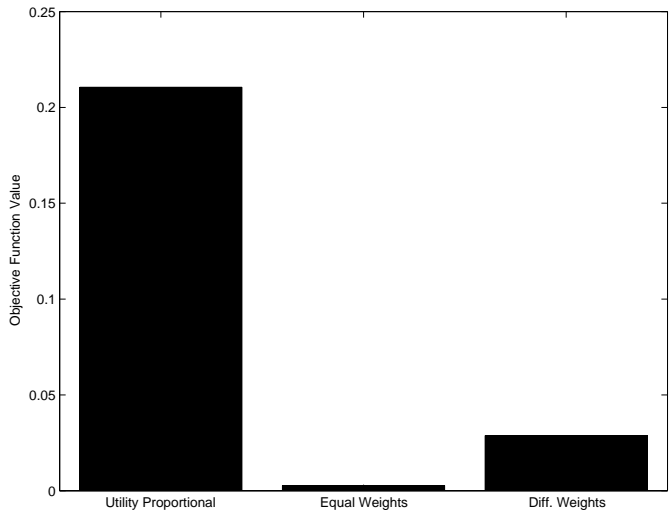}
    \caption{The values of the objective functions with utility proportional and conventional proportional fairness techniques.}
    \label{fig:sim:product_comparison}
\end{figure} 


\section{Conclusion}\label{sec:conclude}
In this paper, we introduced an application-aware resource block scheduling algorithm for LTE systems. The proposed approach assigns resource blocks to the UEs based on the application latency and bandwidth requirements. It gives priority to inelastic traffic compared to delay-tolerant traffic. We first showed that the optimization problem is convex and our scheduling algorithm exists and is optimal. Then, we provided our centralized scheduling algorithm and presented the simulation results. The results show that the QoE is higher with our application-aware approach compared to the proportional fairness approach.   
\bibliographystyle{ieeetr}
\bibliography{pubs}
\end{document}